\documentclass[runningheads]{llncs}
\usepackage{amsmath}
\usepackage{amssymb}
\usepackage{graphicx}
\usepackage{cite}

\newcommand{\NP}{\ensuremath{\mathcal{NP}}}
\def\O(#1){\ensuremath{\mathcal{O}(#1)}}

\title{Path-Additions of Graphs \thanks{Supported by the Deutsche
    Forschungsgemeinschaft (DFG), grant Br835/18-2.}}

\titlerunning{Path-Addition}

\author{Franz J. Brandenburg  \and Alexander Esch \and Daniel Neuwirth}
   \institute{University of Passau, 94030 Passau, Germany \\
   \email{$\{$brandenb, eschalex, neuwirth$\}$@fim.uni-passau.de}}

\authorrunning{F. J. Brandenburg et al.}

\begin{document}

\maketitle

\begin{abstract}
  Path-addition is an operation that takes a graph and adds an internally
  vertex-disjoint path between two vertices together with a set of
  supplementary edges. Path-additions are just the opposite of taking
  minors.

  We show that some classes of graphs are closed under path-addition,
  including non-planar, right angle crossing, fan-crossing free,
  quasi-planar, (aligned) bar 1-visibility, and interval graphs, whereas
  others are not closed, including all subclasses of planar graphs, bounded
  treewidth, $k$-planar, fan-planar, outer-fan planar, outer-fan-crossing
  free, and bar $(1,j)$-visibility graphs.
\end{abstract}

\section{Introduction}

The characterization of   planar graphs by forbidden minors is one
of the highlights of graph theory. It is commonly known as
Kuratowski's theorem \cite{K-cgdp-30} and states that a graph $G$ is
planar if and only if  there is no subgraph that can be obtained
from $K_5$ or $K_{3,3}$ by subdividing edges. In other words, $K_5$
and $K_{3,3}$  are the forbidden topological minors of the planar
graphs \cite{d-gt-00}. Equivalently, there is Wagner's Theorem
\cite{w-minor-37} which states that
 $G$ is planar if and only if  $K_5$ or $K_{3,3}$ cannot be obtained
 from $G$ by edge contractions, edge deletions and the removal of isolated vertices,
 i.e., $K_5$ and $K_{3,3}$ are the minors of $G$.
 An edge
contraction merges the endvertices of an edge  and a subdivision
splits an  edge into a path of length two, i.e.,\ it places a new
vertex on an edge. For more details we refer the reader to Diestel's
textbook \cite{d-gt-00}.

The above characterizations can be turned into a linear-time
planarity testing algorithm which returns a witness: if the tested
graph $G$ is planar, then there is a planar embedding
\cite{mm-ephtpta-96} from which one can obtain a straight-line
planar grid drawing \cite{fpp-hdpgg-90} and if $G$ is non-planar,
then there are distinguished  vertices and vertex-disjoint paths for
the representation of a subdivision of $K_5$ and $K_{3,3}$,
respectively, \cite{w-dfsks-84, mn-LEDA-99}. Thus, the  (extended)
algorithm comes with a  correctness proof.

The theorems of Kuratowski and Wagner are a prominent result in the
theory of graph minors culminating in the  Robertson-Seymour or
Graph Minor Theorem, which states that every  class of graphs
$\mathcal{G}$ closed under taking minors has a finite set of graphs,
called the obstruction set or the set of forbidden minors. A graph
$H$ belongs to $\mathcal{G}$  if and only if $H$ does not contain a
graph in the obstruction set of $\mathcal{G}$ as a minor
\cite{b-pagbt-98, d-gt-00, rs-minorsXX-04}.

The path-addition operation is an extended inverse of edge
contraction and combines subdivision and edge insertion.
Subdivisions    replace  an edge by an internally vertex-disjoint
path. A path-addition starts from scratch and introduces a new path.
Since some classes of graphs are not hereditary, such as chordal
graphs \cite{bls-gc-99, d-gt-00} or graphs with a strong visibility
representation \cite{dett-gdavg-99}, we augment a path-addition by a
set of supplementary edges to meet the requirements of a given class
of graphs. A path-addition can be emulated by an edge insertion,
followed by a subdivision and further edge insertions. Moreover, we
require that an inserted path is long, since too short paths would
violate the defining properties of some classes of graphs. For the
 the supplementary
edges we add the restriction that they are incident to at least one
internal vertex of the path.

Clearly, if a class of graphs $\mathcal{G}$ is closed under taking
minors and under path-addition, then $\mathcal{G}$  is
\emph{trivial}, i.e.,\ $\mathcal{G} = \emptyset$, or $\mathcal{G}$
consists of the empty graph, or $\mathcal{G}$ is the set of all
graphs. In particular, if $\mathcal{G}$ contains a non-empty graph,
then the closure of $\mathcal{G}$ under path-addition includes any
complete graph $K_k$ as a minor. To see this, first add paths until
a graph with at  least $k$ vertices is obtained, and then add vertex
disjoint paths to obtain a subdivision of $K_k$ as a subgraph.

Path-additions of graphs were introduced by Brandenburg et al.\
\cite{ben-ab1v-16}. They are helpful to distinguish classes of
beyond-planar graphs. Beyond-planar is a collective term for classes
of graphs that are defined by restrictions on crossings in visual
representations. Particular examples are 1-planar \cite{ringel-65}
and $k$-planar graphs \cite{pt-gdfce-97}, fan-planar
\cite{bddmpst-fan-15, bcghk-rfpg-14} and fan-crossing free graphs
\cite{cpkk-fan-15}, quasi-planar graphs \cite{aapps-qpg-97}, right
angle crossing (RAC) graphs \cite{del-dgrac-11}, bar
 \cite{DEGLST-bkvg-07}, bar $(1,j)$
\cite{bhkn-bvg-15}, and $1$-visibility graphs \cite{b-vr1pg-14},
rectangle visibility graphs \cite{hsv-rstg-99}, as well as map
graphs \cite{cgp-mg-02, t-mgpt-98}. Some of these classes are
specialized by an alignment of the vertices, such as outerplanar,
outer 1-planar \cite{abbghnr-o1p-15, heklss-ltao1p-15},
outer-fan-planar \cite{bcghk-rfpg-14}, outer-fan-crossing free, and
aligned (or semi) bar 1-visibility \cite{ben-ab1v-16, fm-pbkvg-08}
graphs.

We  summarize our results on the closure of certain classes of
graphs under path-addition in Table \ref{table1}. The closure
properties under subdivision and edge contraction are  known or easy
to obtain. The results show that there is no implication between
path-addition, subdivision, and edge contraction.

\begin{table}
\centering
\begin{tabular}{ l | c |c | c }
    & path- & sub- & edge \\
   graph & addition & division & contraction \\
  \hline
  planar  & -- & + & + \\
  $k$-planar & -- & + & -- \\
  right angle crossing (RAC) & + & + & -- \\
  fan-planar & -- & + & -- \\
  fan-crossing free & + & + & -- \\
  quasi-planar & + & + & -- \\
  bar 1-visibility & + & + & -- \\
  bar $(1,j)$-visibility & -- & + & --\\
  outerplanar & -- & -- & -- \\
  outer 1-planar & -- & -- & -- \\
  outer fan-planar & -- & -- & -- \\
  outer fan-crossing free & -- & -- & -- \\
  aligned bar 1-visibility (AB1V) & + & -- & -- \\
  \hline
\end{tabular}
\vspace{.2cm}
 \caption{Closure properties of classes of graphs}
  \label{table1}
\end{table}

In this work, we study the closure of classes of graphs under
path-addition, where the graphs are defined by a visual
representation. In Section \ref{Sect:prelim} we introduce basic
concepts. A positive closure is studied in Section
\ref{sect:positive} and a negative closure in Section
\ref{sect:negative}.
  We conclude with some open problems.

\section{Preliminaries} \label{Sect:prelim}

We consider simple, undirected graphs $G = (V,E)$ that are defined
by a visual representation. For general graph theoretic terms we
refer to \cite{dett-gdavg-99, d-gt-00}. A class of graphs is the set
of all graph satisfying a particular property. A class of graphs
$\mathcal{G}$ is \emph{hereditary}  if $G \in \mathcal{G}$ implies
that  every induced subgraph of $G$ is in $\mathcal{G}$
\cite{s-egr-03}.

An \emph{embedding} or \emph{drawing}  $\mathcal{E}(G)$ is a mapping
of a graph $G$ into the plane where each vertex is mapped to a point
and each edge $e=(u,v)$ to a Jordan curve connecting the points of
$u$ and $v$. There is a \emph{straight-line drawing} if all curves
of edges are straight lines. Two edges \emph{cross} if their curves
intersect. (Curves of) Edges are not allowed to pass through the
points of other vertices, and edges incident to a vertex do not
cross.

A \emph{generalized visibility representation}  represents each
vertex of a graph by a  horizontal segment, called a \emph{bar},
which is a rectangle of small height and positive width. Each edge
is represented by a  vertical  line of sight with
$k$-\emph{visibility}, such that the line of sight may traverse the
bars of at most $k$ other vertices. Then an edge \emph{crosses} a
vertex. Sometimes, horizontal and vertical are exchanged. Planar
visibility \cite{dett-gdavg-99} with non-transparent bars is
obtained if $k=0$. In an \emph{interval} representation  $k$ is
unbounded such that $k$ overlapping bars or intervals induce $K_k$
as a subgraph.

There are several versions of visibility including \emph{strong},
$\epsilon$, and \emph{weak} visibility. In the strong and
$\epsilon$-versions there is an edge if and only if there is a
visibility. An $\epsilon$-version requires a line of sight of
  width $\epsilon > 0$. This makes a subtle difference, since, in
the planar case,  $K_{2,3}$ is an $\epsilon$-visibility graph but is
not a strong visibility graph. Furthermore, the recognition problem
is \NP-hard for strong visibility graphs and solvable in linear time
for $\epsilon$-visibility graphs \cite{dett-gdavg-99}. Planar
visibility graphs were characterized by Wismath \cite{w-cbg-85} and
Tamassia and Tollis \cite{tt-vrpg-86}. In the weak version there is
a visibility if there is an edge. Hence, edges can be omitted. A
class of weak visibility graphs is hereditary, and the weak
visibility graphs are exactly the induced  subgraphs of the strong
or $\epsilon$-visibility graphs. Also \emph{interval graphs} assume
the strong version of visibility \cite{s-egr-03} and are not
hereditary.

Planarization is a useful tool for embeddings and generalized
visibility representations. It represents the vertices and edges of
the given graph and the edge-edge and edge-vertex crossings. The
\emph{planarization} $\mathcal{P}(\mathcal{E}(G))$ of an embedding
$\mathcal{E}(G)$  is obtained by placing a dummy vertex of degree
four at each crossing point of two edges and thereby subdividing
each crossed edge. Thereafter, $\mathcal{P}(\mathcal{E}(G))$ is an
ordinary planar embedding and there are vertices, edges, and faces,
whose boundary consists of edges and edge segments and is determined
by the vertices and crossing points on the boundary. At each vertex
$v$ (including the dummy vertices) there is a \emph{rotation system}
describing the cycling ordering of the edges incident to $v$ or of
the neighbors of $v$. An embedding is \emph{triangulated} if each
face of the planarization is a triangle.

Similarly, the planarization  $\mathcal{P}(\mathcal{E}(G))$ of a
generalized visibility representation introduces a dummy vertex at
each point on the boundary of a bar or a rectangle, where  a line of
sight either contacts or crosses the bar. In graph drawing these
points are often called ports. Thereafter, the boundary of a bar is
partitioned into segments in counterclockwise order, such that the
boundary of a bar consists of horizontal straight lines and of
orthogonal polylines. Each line of sight is a sequence of straight
vertical segments, which are either in free space or in the interior
of a bar. Each such segment is an edge of the planarization, which
in turn is an orthogonal drawing of a planar graph
\cite{dett-gdavg-99}.
 The planarization of a generalized visibility representation
introduces open and closed faces, where a closed face lies in the
interior of a bar and therefore is excluded for a placement of other
bars or points.

\subsection{Classes of Graphs}
Embeddings and generalized visibility representations are a rich
source for the definition of classes of graphs, and, in particular,
for beyond-planar graphs.
A graph $G$ is $k$-planar \cite{ringel-65,pt-gdfce-97} if it admits
a drawing such that each edge is crossed by at most $k$ other edges,
and is $k$-quasi-planar  \cite{aapps-qpg-97} if there are no $k$
pairwise crossing edges. In a fan-planar drawing
\cite{bddmpst-fan-15, bcghk-rfpg-14} an edge is allowed to cross two
or more edges  if the crossed edges are incident to a vertex and
form a fan, whereas a fan-crossing free drawing \cite{cpkk-fan-15}
excludes edges that cross  two edges  that are incident to the same
vertex. A right angle crossing graph (RAC) \cite{del-dgrac-11}
allows crossings of edges if all edges are represented by straight
lines and edges cross at a right angle. Finally, a map graph
\cite{cgp-mg-02, t-mgpt-98} is obtained from a planar dual. Here
each vertex is represented by a region  and there is an edge if and
only if two regions share at least one point. This results in a
complete subgraph $K_k$ if $k$ regions meet at a single point.

Accordingly, a graph $G$ is a bar $k$-visibility graph
\cite{DEGLST-bkvg-07} if it admits a visibility representation with
$k$-visibility. Then the line of sight of each edge may traverse or
cross at most $k$ other bars. If $k=1$ and each bar is passed by at
most $j$ edges, we obtain bar $(1,j)$-visibility graphs
\cite{bhkn-bvg-15} and $1$-visibility graphs \cite{b-vr1pg-14} if
additionally $j=1$. The two-dimensional extension with
non-transparent rectangles for vertices and horizontal and vertical
lines of sight for edges leads to  rectangle visibility graphs
\cite{hsv-rstg-99}.

Many of these concepts come with a parameter $k$, where $k=0$
corresponds to the planar case  and $k=1$ is most commonly used.
Moreover, there are diverse generalizations, for example from right
angle to large angle drawings \cite{del-dgrac-11}, visibility
representations in 3D \cite{befhlmrrsw-98},  and $L$-shape
visibility \cite{lm-Lvis-16}, see also \cite{l-beyond-14}.

On the other hand, a common specialization comes with an alignment
of the vertices such that they are placed on a line or are attached
to a line. Then all vertices appear in the outer face. So we obtain
outerplanar and outer 1-planar graphs \cite{abbghnr-o1p-15,
heklss-ltao1p-15}, outer-fan-planar \cite{bcghk-rfpg-14} and
outer-fan-crossing free graphs, as well as aligned bar 1-visibility
graphs \cite{ben-ab1v-16}. In the planar case, aligned bar
visibility graphs and  outerplanar graphs coincide
\cite{cdhmm-vrg-95}. Outer 1-planar graphs are planar
\cite{abbghnr-o1p-15} and are a proper subclass of outer fan-planar,
outer fan-crossing free and (weak) aligned bar 1-visibility graphs
\cite{ben-ab1v-16}. Here, each vertex is represented by a vertical
bar with bottom at the $x$-axis and there is an edge if there is a
horizontal line of sight that traverses or crosses at most one other
bar.

\subsection{Path-Addition} \label{sect:path-addition}

A path-addition adds a path of sufficient length between two
vertices and adds a set of supplementary edges which have at least
one end vertex on the path. These are technical restrictions which
help to preserve the properties of  a particular class. The set of
supplementary edges $F$ is not needed if the class of graphs is
hereditary, in which case we let $F = \emptyset$. Otherwise, $F$ is
computed from a   representation.

\begin{definition}
    For a graph  $G = (V,E)$, two vertices  $u, v \in V$, and an
    internally
    vertex-disjoint path
    $P = (u, w_1, \ldots, w_t, v)$ with $w_i \not\in V$ for $1 \leq i \leq
    t$ from $u$ to $v$,
    the  \emph{path-addition} results in a graph
    $G' = (V \cup W, E \cup Q \cup F)$ such that  $W = \{w_1, \ldots,
    w_t\}$ is the set of internal vertices of $P$,
    $Q$ consists of the edges of $P$ and $F$ is a set of supplementary edges with at least
    one endpoint in $W$. We denote $G'$ by $G \oplus P \oplus F$.
\end{definition}

\begin{definition}
    A class of graphs $\mathcal{G}$ is \emph{closed under path-addition} if
    for every graph $G$ in $\mathcal{G}$
    and for every internally  vertex-disjoint path $P$ of length at least
    $|G|-1$ between two
    vertices $u$ and $v$ of $G$ there is a set of edges $F$ such that $G
    \oplus P \oplus F$ is in  $\mathcal{G}$.
    If $\mathcal{G}$ is hereditary, then we let $F
    = \emptyset$.
\end{definition}

Our definition of path-additions comes with two parameters, the
length of the added path and a set of supplementary edges $F$.

We have chosen long paths to be independent of a particular
representation of the graph and to avoid a situation where a too
short path must violate requirements from the given class of graphs.
If paths $P_1, \ldots, P_r$ are added successively to a graph $G$,
then the length of the paths increases at least exponentially such
that $|P_i| \geq |G|-1 + |P_1| + \ldots + |P_{i-1}|$ and $|P_i| \geq
2^{i-1}({G} - 1)+1$.

The  set of supplementary edges $F$ is added to preserve a given
class. This seems necessary if the class of graphs is not closed
under taking subgraphs, e.g., for chordal or triangulated graphs and
for strong visibility representations. The edges are only
constrained by the fact that one endvertex is from the new path.

\section{Positive Closure Results}  \label{sect:positive}

Clearly, a graph $G$ remains non-planar if a path (and  further
edges) are added to $G$. Similarly, non-planarity is preserved by
subdivision, whereas a planar graph may result from a non-planar one
by an edge contraction or an edge   or a vertex removal.

\begin{corollary} \label{cor:non-planar}
The class of non-planar graphs is closed under path-addition.
\end{corollary}

Many classes of graphs admit the routing of a vertex-disjoint path
along a simple or shortest path between two vertices. Details are
obtained from the visual representation. This technique is
applicable if a given edge or vertex can be crossed by a new edge
between two new vertices.

\begin{theorem}
The following classes of graphs are closed under path-addition:
\begin{itemize}
  \item right angle crossing graphs (RAC)
  \item fan-crossing free graphs
  \item quasi-planar graphs
  \item bar 1-visibility graphs
  \item aligned bar 1-visibility graphs (AB1V)
  \item interval graphs
\end{itemize}
\end{theorem}

\begin{proof}
Consider a graph  $G$ with  vertices $u$ and $v$   and an internally
vertex-disjoint  path $P = (u_0, \ldots, u_r)$ from $u$ to $v$. Let
 $S = (v_0,\ldots, v_s)$ be a shortest path from $u$ to $v$ in $G$.
 Thus $u = u_0 = v_0$ and $v = u_0 = v_0$.

First, we consider RAC, fan-crossing free and quasi-planar graphs,
where a graph $G$ is given by an embedding $\mathcal{E}(G)$. Since
these classes are hereditary, the set of supplementary edges is
empty. Path $P$ must be added to the embedding, and we route $P$
along $S$ in $\mathcal{E}(G)$.  Consider the $i$-th vertex $v_i$ of
$S$. If $(e_1, \ldots, e_t)$ is the rotation system at vertex $v_i$
given by $\mathcal{E}(G)$, then $S$ enters $v_i$ via edge $e_h$ and
leaves $v_i$ via  $e_j$. Suppose that $h-j \leq j-h$ in the circular
ordering of the rotation system. The other case is similar. Then
route a section of length $j-h-1$ of $P$ around   $v_i$. We
associate this section of $P$ with $v_i$. If vertices $u_1, \ldots
u_q$ have been associated with the vertices $v_0,\ldots, v_{i-1}$,
then associate $u_q+1, \ldots, u_q + h-j$ to $v_i$ and place these
vertices close to $v_i$ such that the edge $(u_{q + \nu}, u_{q+
\nu+1})$ crosses $e_{h+\nu}$ for $\nu=1, \ldots, h-j-1$, see Fig.\
\ref{fig:atvertex}. Note that the first edge $(u_{q +1}, u_{q+2})$
crosses edge $(v_{i-1}, v_i)$ of $S$ if $j-h < h-j$.  If edge $(v_i,
v_{i+1})$ of $S$ crosses   edge $f$ in $\mathcal{E}(G)$, then edge
$(u_{\mu}, u_{\mu+1})$ of $P$ crosses $f$, where $u_{\mu}$ is the
last  vertex of $S$ associated with $v_i$  and $u_{\mu+1}$ is the
first   vertex associated with $v_{i+1}$. If $\mathcal{E}(G)$ is a
RAC drawing, then the  vertices associated with each $v_i$  are
placed such that the edges $(u_{q +\nu}, u_{q+\nu+1})$ and
$e_{h+\nu}$ for $\nu=1, \ldots, h-j-1 $ cross at a right angle. Edge
$(u_{\mu}, u_{\mu+1})$ is parallel to the edge $(v_i, v_{i+1})$ and
if $(v_i, v_{i+1})$ crosses some edge $f$ at a right angle, then so
does $(u_{\mu}, u_{\mu+1})$. Clearly, the edges of $P$ cannot
introduce a fan-crossing.

If $P$ is longer than the sum of the associated segments, then
insert the remaining subpath just before $v$.

It remains to show that a length of $n-1$ suffices for the routing
of $P$ in $\mathcal{E}(G)$.
A RAC graph has at most $4n-10$ edges \cite{del-dgrac-11}. Then $G$
is maximal and no further edge can be added. So assume that $G$ is
maximal. If a segment of length $\lambda$ of $P$ is associated with
vertex $v_i$, then at least $4 \lambda$ edges of $G$ can be
associated with the segment. Each edge $(u_{q +\nu}, u_{q+\nu+1})$
crosses an  edge $e_{h+\nu} = (v_i, w_\nu)$  and by the maximality
of $G$ there is an edge $(w_{\nu}, w_{\nu+1})$ which is not crossed
by any edge of $P$ since $S$ is a shortest path. Hence, we can
associate the edges incident to $v_i$ together with the edges
connecting consecutive endvertices according to the rotation system
with the segment of $P$ associated with $v_i$, and this is a $4:1$
relation. In consequence, a path $P$ of length at most $(4n-10)/4$
can be routed between vertices $u$ and $v$ in the RAC   embedding
$\mathcal{E}(G)$. The same argument applies to fan-crossing free
graphs with at most $4n-8$ edges \cite{cpkk-fan-15}.

For quasi-planar graphs, we start with a single edge $e$ for $P$ and
follow $S$ from $u$ to $v$. If there is a violation of
quasi-planarity and $e$ crosses two other crossing edges $f$ and
$g$, then we subdivide $e$ with a vertex $u$ between $e \times f$
and $e \times g$, where $e \times f$  is the crossing point  of $e$
and $f$. We account $u$ to the other endvertex of $f$, which is not
passed by $S$ and $P$ if $S$ is a shortest path. Hence, $P$ does not
introduce three mutually crossing edges and a length of $n-1$
suffices for $P$.

Next, we consider generalized visibility representations.
%
%
In a bar 1-visibility representation with horizontal bars shoot a
vertical ray from vertex $u$ to the top. Suppose the ray traverses
bars $b_1, \ldots, b_r$. Then introduce a small bar $b'_i$ of width
$\epsilon$ just above $b_i$ such that $b'_i$ is not traversed by any
other line of sight, whereas $(b'_{i-1}, b'_i)$ traverses $b_i$,
where $b_0$ is the bar of $u$. Proceed similarly for a path from
$v$. Introduce a  new topmost bar between the rays. The length of
the new path from $u$ to $v$ is at most $n-1$ if the sum of the
distances of $u$ and $v$ to the top is at most the distance to the
bottom. Otherwise, we go to the bottom.

In case of a strong bar 1-visibility, one must add further edges to
meet the requirements. These edges can be retrieved from the bar
$1$-visibility representation.

In an aligned bar 1-visibility representation, the bars of inner
vertices of the path $P$ are shorter than the bars of the given
vertices and they are inserted in alternation with bars of vertices
of $G$, as elaborated in   \cite{ben-ab1v-16}.

Finally, in an interval representation, join the intervals of $u$
and $v$ by a sequence of overlapping intervals, and add
supplementary edges to preserve interval graphs.
 \qed
\end{proof}

\begin{figure}
    \centering
    \includegraphics[scale =0.7]{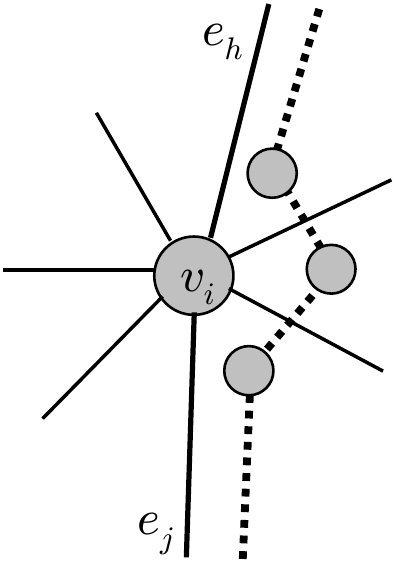}
    \caption{Routing a new path with dotted edges at a vertex in parallel with a given path
     whose edges are drawn bold}
    \label{fig:atvertex}
\end{figure}

There are many other classes of graphs \cite{bls-gc-99, s-egr-03},
and it seems likely that some are closed under path-addition, e.g.,
rectangle visibility graphs.

\section{Negative Closure Results}  \label{sect:negative}

As aforesaid, the closure under path-addition admits the
construction of every complete graph $K_k$ as a minor. In
consequence, path-addition opposes taking minors and can be regarded
as an \emph{anti-minor}.

\begin{corollary} \label{cor:minors}
A minor closed class of graphs $\mathcal{G}$ is closed under
path-addition if and only if $\mathcal{G}$ is trivial, i.e.,
$\mathcal{G}$ consists of the empty set, $\mathcal{G}$ contains only
the empty graph, or $\mathcal{G}$ is the set of all graphs.
\end{corollary}

Forthcoming, we shall exclude trivial classes of graphs.

As an immediate consequence of the construction of $K_k$ minors for
arbitrary $k>0 $  we obtain that classes of graphs with bounded
treewidth and all subclasses of the planar graphs are destroyed by
path-addition.

\begin{corollary} \label{cor:treewidth}
A class of graphs  $\mathcal{G}$  is not closed under path-addition
if $\mathcal{G}$ has bounded treewidth.
\end{corollary}

\begin{corollary}
If $\mathcal{G}$ is a subclass of  planar graphs, then $\mathcal{G}$
is not closed under path-addition. In particular, the outerplanar,
outer 1-planar, series-parallel, and  planar graphs are not closed
under path-addition.
\end{corollary}

For further non-closure results we use the following technique:
Suppose there is a cycle $C$ of length $r$ separating a graph $G$
into an inner and an outer component such that the components are
nonempty and $C$ can be traversed only  $c \cdot k$ times, where $c$
is the length of $C$. Then we can add $c \cdot k +1$ paths to
violate this property.

There are some   obstacles for the application of this technique.
First, a graph in many classes of beyond-planar graphs does not
necessarily have a unique embedding or visibility representation. In
consequence, two vertices $u$ and $v$ may be separated by a cycle
$C$ in one embedding whereas $C$ does not separate them in another
embedding. Second, the removal of the cycle does not necessarily
partition a graph into components, since   connectivity may be
preserved by crossing edges. Therefore, we consider a planarization.
The removal of a vertex $v$ from the planarization induces the
removal of the vertex and all edges of the given graph that are
incident to $v$. In particular, if $v$ is a crossing point of two
edges then both edges are removed, and similarly for an edge-vertex
crossing in visibility representations where the edge and the
crossed vertex are removed.

\begin{definition}
A graph  $G$  has the  \emph{separating cycle property}, SCP for
short, if for every planarization $\mathcal{P}(\mathcal{E}(G))$ of
an embedding or a generalized visibility representation
 there is a cycle $C = (v_0, e_1, v_1, \ldots, e_r, v_r)$  of consecutive vertices
 and edges such that  $G-C$ decomposes into $s > 1$ components $C_1, \ldots, C_s$ and there are at least two
components $C_i$ and $C_j$ with $i \neq j$ containing  a vertex of
$G$.
\end{definition}

If each edge or vertex of $C$ can be traversed at most $k$ times,
then there is an upper bound on the number of path-additions between
two vertices in distinct components.

\begin{lemma}
If a class of graphs $\mathcal{G}$ contains a graph $G$ with SCP and
every edge or vertex of $G$ can be crossed at most $k$ times, then
$\mathcal{G}$ is not closed under path-addition.
\end{lemma}
\begin{proof}
Consider two vertices $u$ and $v$ in different components of $G-C$,
and add at least $c \cdot k +1$ paths between $u$ and $v$, where $c$
is the length of $C$. Each path must traverse $C$ which allows at
most $c \cdot k$ traversals.
\qed
\end{proof}

In consequence, we obtain:

\begin{corollary}
The 1-planar graphs and the $(1,j)$ bar 1-visibility graphs for $j
\leq 4$ satisfy SCP and are not  closed under path-addition.
\end{corollary}
 \begin{proof}
 The extended wheel graphs $XW_{2k}$  are 1-planar graphs
with two poles  $p$ and $q$ and a cycle of $2k$ vertices such that
$p$ is inside and $q$ is outside the cycle, or vice-versa
\cite{s-o1p-84, s-rm1pg-10}. Since each edge of the cycle can be
traversed at most once, there are at most $2k$ path-additions for
$p$ and $q$. In fact, an extended wheel graph does not admit any
further path between its poles.

 Similarly, there are
bar $(1,j)$-visibility graphs for $1 \leq j \leq 4$ with a fixed bar
1-visibility representation \cite{bhkn-bvg-15} which satisfy SCP and
each bar can be traversed at most $j$ times. \qed
\end{proof}

In  passing, we note that 4-map graphs are not closed under
path-addition, since they are the triangulated 1-planar graphs
\cite{b-4m1pg-15, cgp-rh4mg-06}.

\begin{lemma} \label{lem:K2qcouterex}
For every $k \geq 0$ there is a  complete bipartite graph  $K_{2,q}$
 which satisfies SCP for  $k$-planar and bar
$(1,k)$-visibility graphs.
\end{lemma}
\begin{proof}
For $k$-planarity, let $q \geq 4k+9$  and,  towards a contradiction,
assume that  $K_{2,q}$ does not satisfy SCP. Let $\{u_1, u_2\}$ and
$\{v_1,\ldots, v_q\}$ be the sets of vertices of $K_{2,q}$. Suppose
there is an initial planar quadrangle $(u_1, v_1, u_2, v_2)$,
otherwise, the edges $(u_1, v_i)$ and $(u_2, v_1)$ cross for $i=
2,\ldots, q$ in the rotation system at $u_1$, which violates
$k$-planarity. To violate SCP, all vertices $v_3, \ldots, v_q$ must
lie in one face of the quadrangle $(u_1, v_1, u_2, v_2)$ and there
is no other quadrangle as a separating cycle. Suppose all of $v_3,
\ldots, v_q$ lie in the outer face. Then each pair of edges incident
to $v_i$ crosses one edge of the quadrangle for $i=3, \ldots, q$,
which violates $k$-planarity for $q \geq 4k+9$.

The same argument applies to bar $(1,k)$-visibility representations,
where the edges incident to $v_i$ for $i=3, \ldots, q$ traverse one
of the bars of the vertices from the initial quadrangle $(u_1, v_1,
u_2, v_2)$.
\qed
\end{proof}

\begin{corollary} \label{cor:barvisnotpath}
The $k$-planar graphs and the bar $(1,k)$ visibility graphs are not
closed under path-addition for every $k \geq 0$.
\end{corollary}

It is not immediately clear that fan-planar graphs
\cite{ku-dfang-14, bddmpst-fan-15, bcghk-rfpg-14} satisfy SCP, since
there is no upper bound on the number of crossings per edge.
However, Proposition 1 of Kaufmann and Ueckerdt \cite{ku-dfang-14}
comes close. It states that a connected planar graph can be expended
to a maximal fan-planar graph preserving the planar edges. As an
example consider the crossed dodecahedral graph from
\cite{ku-dfang-14}. The dodecahedral graph is a planar 3-regular
graph of the dodecahedron with 20 vertices, 30 edges, and 12
pentagonal faces. In the crossed version there is a $K_5$ for each
face.

\begin{lemma} \label{lem:fanplanar}
The fan-planar graphs ar not closed under path-addition.
\end{lemma}

\begin{proof}
We claim that the dodecahedral graph with a crossed pentagram in
each face has a unique fan-planar embedding up to graph automorphism
\cite{s-rm1pg-10}. Towards a contradiction, suppose there is another
embedding. Then some $K_5$ subgraphs must be embedded with at least
one pair of crossing edges and at least one edge $(v_i, v_{i+2})$ in
the outer face, where $(v_1, \ldots, v_5)$ is the cyclic ordering of
the vertices. However, there are two  $K_5$   with vertices $v_i,
v_{i+1},  w, z_1, z_2$ and $v_{i+1}, v_{i+2}, w, z_3, z_4$ such the
edges $(v_{i+1}, z_i)$ for $i=1,2,3,4$ do not admit a fan-planar
embedding, since they must cross $(v_i, v_{i+1})$. In addition, a
crossed pentagram is impenetrable for any vertex-disjoint path,
since all edges are crossed by a fan. Hence, one cannot add an
internally vertex-disjoint path from a vertex of the inner face to a
vertex of the outer face of the underlying planar pentagram.
\qed
\end{proof}

\begin{lemma} \label{lem:non-outer-fan}
The outer-fan-planar and the outer-fan-crossing free graphs are not
closed under path-addition.
\end{lemma}
\begin{proof}
Consider $K_5$, which is outer-fan-planar, and then add paths from
one vertex  $v_1$ to all other vertices. The obtained graph $G$ is
biconnected  and, if outer-fan-planar, would admit a straight-line
outer-fan-planar drawing with all vertices on a circle, as proved in
\cite{bcghk-rfpg-14}, Lemma 1. However, this is impossible.

Similarly, $K_4$ is outer-fan-crossing-free and the addition of a
path for each pair of vertices violates this property.
\qed
\end{proof}

At last, we consider fan-planar graphs. Binucci et al.
\cite{bddmpst-fan-15} have proved that there are 2-planar graphs
that are not fan-planar. The graph is illustrated in Fig.\
\ref{fig:nonfanplanar}, where each bold line is replaced by a
fan-planar embedding of $K_7$ and each thin or line  represents an
edge.

\begin{lemma} \label{lem:non-outer-fan}
The fan-planar  graphs are not closed under path-addition.
\end{lemma}
\begin{proof}
Consider  graph $G$ from   Binucci et al. \cite{bddmpst-fan-15} and
remove the edges $(v_2, v_6)$ and $(v_3, v_9)$. The so-obtained
graph is fan-planar, where each bold line represents a fan-planar
embedding of $K_7$. Replace the edges $(v_2, v_6)$ and $(v_3, v_9)$
by vertex-disjoint paths $P_1$ and $P_2$. Binucci et al. argue that
the edges $(v_2, v_6)$ and $(v_3, v_9)$ must be routed on one side
of the cycle $v_1, \ldots, v_{10}$, and this also holds for $P_1$
and $P_2$, which, inevitably, introduces a crossing of $(v_1, v_7)$
by two independent edges, and similarly for $(v_4,v_8)$.
\qed
\end{proof}

\begin{figure}
    \centering
    \includegraphics[scale=0.4]{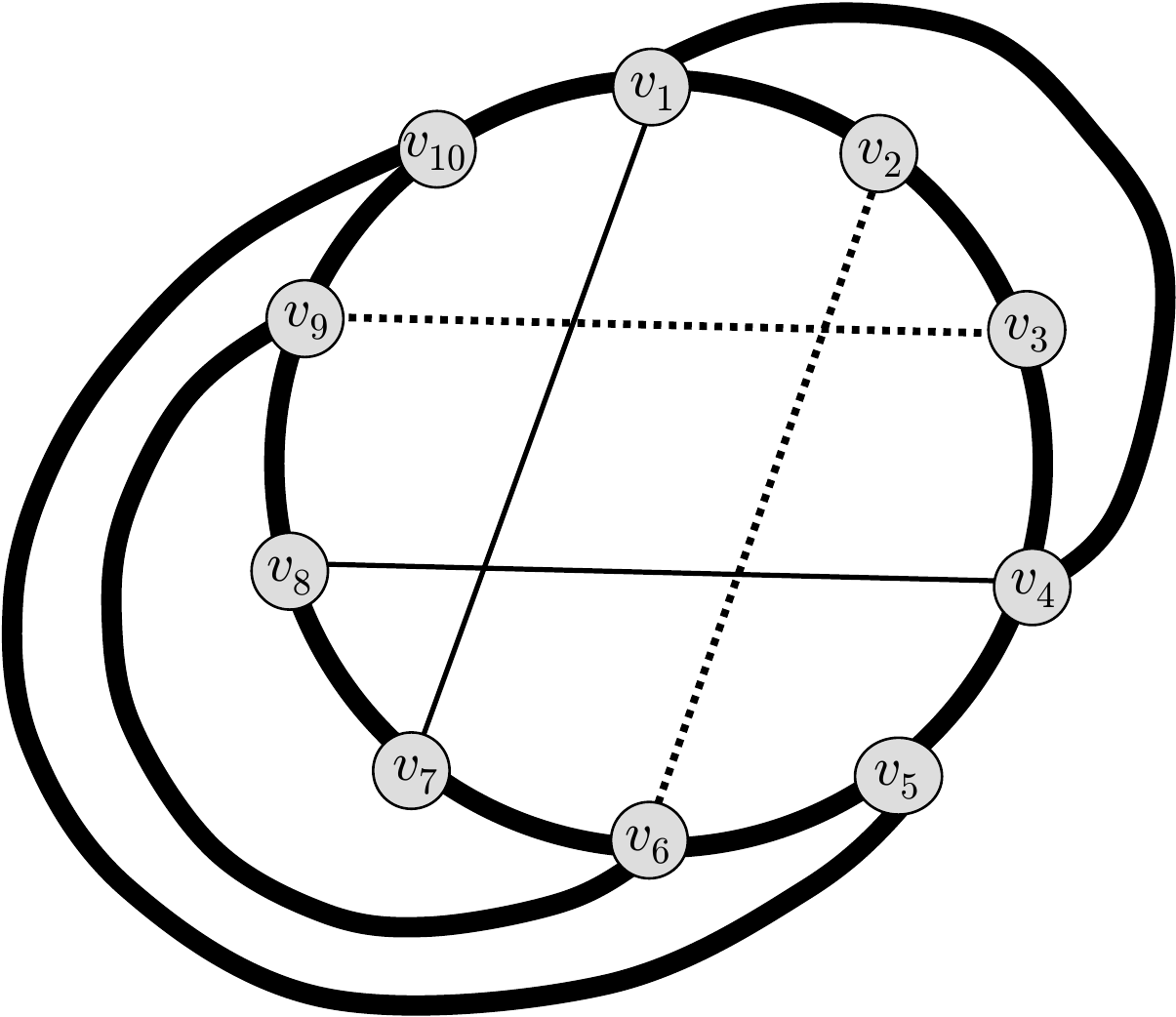}
    \caption{A non-fan-planar graph with bold edges representing a fan-planar embedding of $K_7$ from \cite{bddmpst-fan-15}}
    \label{fig:nonfanplanar}
\end{figure}

\section{Conclusion}

In this work we have investigated the path-addition operation which is
opposite of taking minors and can be emulated by edge addition and
subdivision. We have shown that some classes of graphs are closed under
path-addition and others are not. It might be worthwhile to investigate the
closure of classes of graphs under path-addition. For example, a graph that
is obtained from a planar graph by path-addition is a RAC graph, and a graph
that is obtained from a 1-planar graph is a fan-crossing free and a bar
1-visibility graph. Which graphs are RAC graphs and cannot be obtained by
path-addition from planar graphs, even with a relaxed version of
path-addition without a length restriction?


\bibliographystyle{abbrv}
\bibliography{brandybibV3}

\begin{thebibliography}{10}

\bibitem{aapps-qpg-97}
P.~K. Agarwal, B.~Aronov, J.~Pach, R.~Pollack, and M.~Sharir.
\newblock Quasi-planar graphs have a linear number of edges.
\newblock {\em Combinatorica}, 17(1):1--9, 1997.

\bibitem{abbghnr-o1p-15}
C.~Auer, C.~Bachmaier, F.~J. Brandenburg, A.~Glei\ss{}ner, K.~Hanauer,
  D.~Neuwirth, and J.~Reislhuber.
\newblock Outer 1-planar graphs.
\newblock {\em Algorithmica}, 74(4):1293--1320, 2016.

\bibitem{bcghk-rfpg-14}
M.~A. Bekos, S.~Cornelsen, L.~Grilli, S.~Hong, and M.~Kaufmann.
\newblock On the recognition of fan-planar and maximal outer-fan-planar graphs.
\newblock In C.~A. Duncan and A.~Symvonis, editors, {\em {GD} 2014}, volume
  8871 of {\em {LNCS}}, pages 198--209. Springer, 2014.

\bibitem{bddmpst-fan-15}
C.~Binucci, E.~{Di Giacomo}, W.~Didimo, F.~Montecchiani, M.~Patrignani,
  A.~Symvonis, and I.~G. Tollis.
\newblock Fan-planarity: Properties and complexity.
\newblock {\em Theor. Comput. Sci.}, 589:76--86, 2015.

\bibitem{b-pagbt-98}
H.~L. Bodlaender.
\newblock A partial \emph{k}-arboretum of graphs with bounded treewidth.
\newblock {\em Theor. Comput. Sci.}, 209(1-2):1--45, 1998.

\bibitem{befhlmrrsw-98}
P.~Bose, H.~Everett, S.~P. Fekete, M.~E. Houle, A.~Lubiw, H.~Meijer,
  K.~Romanik, G.~Rote, T.~C. Shermer, S.~Whitesides, and C.~Zelle.
\newblock A visibility representation for graphs in three dimensions.
\newblock {\em J. Graph Algorithms Appl.}, 2(2), 1998.

\bibitem{b-vr1pg-14}
F.~J. Brandenburg.
\newblock 1-visibility representation of 1-planar graphs.
\newblock {\em J. Graph Algorithms Appl.}, 18(3):421--438, 2014.

\bibitem{b-4m1pg-15}
F.~J. Brandenburg.
\newblock On 4-map graphs and 1-planar graphs and their recognition problem.
\newblock Technical Report {arXiv:1509.03447 [cs.CG]}, Computing Research
  Repository ({CoRR}), February 2015.
\newblock \url{http://arxiv.org/abs/1509.03447}.

\bibitem{ben-ab1v-16}
F.~J. Brandenburg, A.~Esch, and D.~Neuwirth.
\newblock On aligned bar 1-visibility graphs.
\newblock In M.~Kaykobad and R.~Petreschi, editors, {\em {WALCOM} 2016}, pages
  95--106, 2016.

\bibitem{bhkn-bvg-15}
F.~J. Brandenburg, N.~Heinsohn, M.~Kaufmann, and D.~Neuwirth.
\newblock On bar (1, j)-visibility graphs - (extended abstract).
\newblock In M.~S. Rahman and E.~Tomita, editors, {\em WALCOM 2015}, volume
  8973 of {\em LNCS}, pages 246--257. Springer, 2015.

\bibitem{bls-gc-99}
A.~Brandst{\"a}dt, V.~B. Le, and J.~P. Spinrad.
\newblock {\em Graph Classes: A Survey}.
\newblock SIAM Monographs on Discrete Mathmeatics and Applications. SIAM, 1999.

\bibitem{cgp-mg-02}
Z.~Chen, M.~Grigni, and C.~H. Papadimitriou.
\newblock Map graphs.
\newblock {\em J. {ACM}}, 49(2):127--138, 2002.

\bibitem{cgp-rh4mg-06}
Z.~Chen, M.~Grigni, and C.~H. Papadimitriou.
\newblock Recognizing hole-free 4-map graphs in cubic time.
\newblock {\em Algorithmica}, 45(2):227--262, 2006.

\bibitem{cpkk-fan-15}
O.~Cheong, S.~Har{-}Peled, H.~Kim, and H.~Kim.
\newblock On the number of edges of fan-crossing free graphs.
\newblock {\em Algorithmica}, 73(4):673--695, 2015.

\bibitem{fpp-hdpgg-90}
H.~de~Fraysseix, J.~Pach, and R.~Pollack.
\newblock How to draw a planar graph on a grid.
\newblock {\em Combinatorica}, 10:41--51, 1990.

\bibitem{DEGLST-bkvg-07}
A.~M. Dean, W.~Evans, E.~Gethner, J.~D. Laison, M.~A. Safari, and W.~T.
  Trotter.
\newblock Bar k-visibility graphs.
\newblock {\em J. Graph Algorithms Appl.}, 11(1):45--59, 2007.

\bibitem{dett-gdavg-99}
G.~Di~Battista, P.~Eades, R.~Tamassia, and I.~G. Tollis.
\newblock {\em Graph Drawing: Algorithms for the Visualization of Graphs}.
\newblock Prentice Hall, 1999.

\bibitem{del-dgrac-11}
W.~Didimo, P.~Eades, and G.~Liotta.
\newblock Drawing graphs with right angle crossings.
\newblock {\em Theor. Comput. Sci.}, 412(39):5156--5166, 2011.

\bibitem{d-gt-00}
R.~Diestel.
\newblock {\em Graph Theory}.
\newblock Springer, 2000.

\bibitem{fm-pbkvg-08}
S.~Felsner and M.~Massow.
\newblock Parameters of bar k-visibility graphs.
\newblock {\em J. Graph Algorithms Appl.}, 12(1):5--27, 2008.

\bibitem{cdhmm-vrg-95}
F.J.Cobos, J.~Dana, F.~Hurtado, A.~Márquez, and F.~Mateos.
\newblock On a visibility representation of graphs.
\newblock In F.~J. Brandenburg, editor, {\em Graph Drawing}, volume 1027 of
  {\em LNCS}, pages 152--161. Springer, 1995.

\bibitem{heklss-ltao1p-15}
S.~Hong, P.~Eades, N.~Katoh, G.~Liotta, P.~Schweitzer, and Y.~Suzuki.
\newblock A linear-time algorithm for testing outer-1-planarity.
\newblock {\em Algorithmica}, 72(4):1033--1054, 2015.

\bibitem{hsv-rstg-99}
J.~P. Hutchinson, T.~Shermer, and A.~Vince.
\newblock On representations of some thickness-two graphs.
\newblock {\em Computational Geometry}, 13:161--171, 1999.

\bibitem{ku-dfang-14}
M.~Kaufmann and T.~Ueckerdt.
\newblock The density of fan-planar graphs.
\newblock Technical Report arXiv:1403.6184 [cs.DM], Computing Research
  Repository ({CoRR}), March 2014.
\newblock \url{http://arxiv.org/abs/1403.6184}.

\bibitem{K-cgdp-30}
K.~Kuratowski.
\newblock Sur le probl$\acute{e}$me des courbes gauches en topologie.
\newblock {\em Fund. Math.}, 15:271--283, 1930.

\bibitem{l-beyond-14}
G.~Liotta.
\newblock Graph drawing beyond planarity: some results and open problems.
\newblock In S.~Bistarelli and A.~Formisano, editors, {\em Proc. 15th Italian
  Conference on Theoretical Computer Science}, volume 1231 of {\em {CEUR}
  Workshop Proceedings}, pages 3--8. CEUR-WS.org, 2014.

\bibitem{lm-Lvis-16}
G.~Liotta and F.~Montecchiani.
\newblock L-visibility drawings of {IC}-planar graphs.
\newblock {\em Inf. Process. Lett.}, 116(3):217--222, 2016.

\bibitem{mm-ephtpta-96}
K.~Mehlhorn and P.~Mutzel.
\newblock On the embedding phase of the {H}opcroft and {T}arjan planarity
  testing algorithm.
\newblock {\em Algorithmica}, 16(2):233--242, 1996.

\bibitem{mn-LEDA-99}
K.~Mehlhorn and S.~N{\"{a}}her.
\newblock {\em {LEDA:} {A} Platform for Combinatorial and Geometric Computing}.
\newblock Cambridge University Press, 1999.

\bibitem{pt-gdfce-97}
J.~Pach and G.~T{\'o}th.
\newblock Graphs drawn with a few crossings per edge.
\newblock {\em Combinatorica}, 17:427--439, 1997.

\bibitem{ringel-65}
G.~Ringel.
\newblock Ein {S}echsfarbenproblem auf der {K}ugel.
\newblock {\em Abh. aus dem Math. Seminar der Univ. Hamburg}, 29:107--117,
  1965.

\bibitem{rs-minorsXX-04}
N.~Robertson and P.~D. Seymour.
\newblock Graph minors {XX}. {W}agner's conjecture.
\newblock {\em J. Comb. Theory, Ser. {B}}, 92(2):325--357, 2004.

\bibitem{s-o1p-84}
H.~Schumacher.
\newblock {\"U}ber 1-optimale {G}raphen.
\newblock {\em Mathematische Nachrichten}, 117:323--339, 1984.

\bibitem{s-egr-03}
J.~P. Spinrad.
\newblock {\em Efficient Graph Representations}.
\newblock Fields Institute Monographs 19. American Mathematical Society,
  Providence, 2003.

\bibitem{s-rm1pg-10}
Y.~Suzuki.
\newblock Re-embeddings of maximum 1-planar graphs.
\newblock {\em {SIAM} J. Discr. Math.}, 24(4):1527--1540, 2010.

\bibitem{tt-vrpg-86}
R.~Tamassia and I.~G. Tollis.
\newblock A unified approach a visibility representation of planar graphs.
\newblock {\em Discrete Comput. Geom.}, 1:321--341, 1986.

\bibitem{t-mgpt-98}
M.~Thorup.
\newblock Map graphs in polynomial time.
\newblock In {\em Proc. 39th {FOCS}}, pages 396--405. {IEEE} Computer Society,
  1998.

\bibitem{w-minor-37}
K.~Wagner.
\newblock {\"U}ber eine {E}igenschaft der ebenen {K}omplexe.
\newblock {\em Math. {A}nn.}, 114:570--–590, 1937.

\bibitem{w-dfsks-84}
S.~G. Williamson.
\newblock Depth-first search and {K}uratowski subgraphs.
\newblock {\em J. of the ACM}, 31(4):681--693, 1984.

\bibitem{w-cbg-85}
S.~Wismath.
\newblock Characterizing bar line-of-sight graphs.
\newblock In {\em Proc. 1st ACM Symp. Comput. Geom.}, pages 147--152. ACM
  Press, 1985.

\end{thebibliography}

\end{document}